\documentclass[12pt,a4paper]{article}

\usepackage{amsmath,amsthm,amssymb,amscd,a4wide}

\usepackage{dsfont}
\usepackage{mathrsfs}
\usepackage{setspace}
\usepackage{hyperref}

\newcommand{\ud}{\mathrm{d}}

\newcommand{\cD}{{\mathcal D}}

\newtheorem{theorem}{Theorem}

\newtheorem{prop}{Proposition}

\newtheorem{remark}{Remark}

\theoremstyle{definition}

\begin{document}

\thispagestyle{empty}

\vspace*{1cm}

\begin{center}

{\LARGE\bf On the number of isolated eigenvalues of a pair of particles in a quantum wire} \\

\vspace*{2cm}

{\large Joachim Kerner \footnote{E-mail address: {\tt Joachim.Kerner@fernuni-hagen.de}} }%

\vspace*{5mm}

Department of Mathematics and Computer Science\\
FernUniversit\"{a}t in Hagen\\
58084 Hagen\\
Germany\\

\end{center}

\vfill

\begin{abstract} In this note we consider a pair of particles moving on the positive half-line $\mathbb{R}_+$ with the pairing generated by a hard-wall potential. This model was first introduced in \cite{KernerBoundSystem} and later applied to investigate condensation of pairs of electrons in a quantum wire \cite{KernerBoundROMP,KernerPairsElectrons}. For this, a detailed spectral analysis proved necessary and as a part of this it was shown in \cite{KernerBoundROMP} that, in a special case, the discrete spectrum of the Hamiltonian consists of a single eigenvalue only. It is the aim of this note to prove that this is generally the case. 
\end{abstract}

\newpage

\section{Introduction}
In this note we consider an interacting system of two particles with the positive half-line $\mathbb{R}_+=(0,\infty)$ as one-particle configuration space. More explicitly, the Hamiltonian shall be given by 
\begin{equation}\label{GenHamil}
H=-\frac{\partial^2}{\partial x^2}-\frac{\partial^2}{\partial y^2}+v(|x-y|)
\end{equation}
with hard-wall potential, $d > 0$, 
\begin{equation}\label{Potential}
v(x):=\begin{cases}
0 \quad x < d\ , \\
\infty \quad \text{else}\ .
\end{cases}
\end{equation}
Note that, through the potential $v$, the two particles actually form a pair with spatial extension characterised by $d> 0$. 

The two-particle model with Hamiltonian~\eqref{GenHamil} and potential~\eqref{Potential} was introduced in \cite{KernerBoundSystem}. Its investigation grew out of studying many-particle quantum chaos on quantum graphs \cite{BKSingular,BKContact} taking into account recent results in theoretical physics \cite{Queisserunruh}. More generally, due to the technical advances in the last decades and especially in the realm on nanotechnology, it has become pivotal to study the properties of interacting particle systems in one dimension which may differ greatly from those of systems in higher dimension \cite{ThierryGamarchi,GIAMARCHI2016322}. Also, since the pairing of electrons (Cooper pairs) in metals is the key mechanism in the formation of the superconducting phase in type-I superconductors \cite{CooperBoundElectron,BCSI}, an investigation of the Hamiltonian~\eqref{GenHamil} is also interesting from a solid-state physics point of view. And indeed, the condensation of pairs of electrons with Hamiltonian~\eqref{GenHamil} was studied in \cite{KernerBoundROMP,KernerPairsElectrons}; in \cite{KernerBoundROMP} the electrons forming a pair have same spin and in \cite{KernerPairsElectrons} the electrons of a pair have opposite spin as it is the case with Cooper pairs. 

In this note we are interested in spectral properties of the $H$. More explicitly, we are interested in characterising the discrete part of the spectrum. It was the key observation in \cite{KernerBoundSystem} that the discrete spectrum of $H$ is non-trivial, i.e., there exist eigenvalues below the botton of the essential spectrum. Since this is not the case if one changes the one-particle configuration space to be the whole real line $\mathbb{R}$, the existence of a discrete spectrum is directly linked to the geometry of the one- and two-particle configuration space. Implementing exchange symmetry, the authors of \cite{KernerBoundROMP} were able to show that the discrete spectrum actually exists of one eigenvalue only. The main purpose of this note is to show that exchange symmetry is indeed not necessary and that the discrete spectrum always consists of one eigenvalue only.

Finally, we want to draw attention to the recent paper \cite{KernerEggerPankrashkin} in which spectral properties of \eqref{GenHamil} for a large class of interaction potentials $v:\mathbb{R}_+ \rightarrow \mathbb{R}$ where studied. The authors also found that the discrete spectrum is non-empty and contains only finitely many eigenvalues. However, no bounds on the number of isolated eigenvalues were derived. 

\section{The model and main results}
Due to the formal nature of the interaction potential \eqref{Potential}, $H$ cannot be directly realised as a self-adjoint operator on $L^2(\mathbb{R}^2_+)$. However, we see that this choice for $v$ means that the two-particle configuration space is actually given by
\begin{equation}
\Omega=\{(x,y) \in \mathbb{R}^2_+:\ |x-y| < d   \}\ .
\end{equation}
Based on $\Omega$ we then introduce the Hilbert spaces $L^2(\Omega):=L^2_{0}(\Omega)$ as well as
\begin{equation}\begin{split}
L^2_s(\Omega)&:=\{\varphi \in L^2(\Omega): \varphi(x,y)=\varphi(y,x)  \}\ ,\\
L^2_a(\Omega)&:=\{\varphi \in L^2(\Omega): \varphi(x,y)=-\varphi(y,x)  \}\ .
\end{split}
\end{equation}
When describing two distinguishable particles one focusses on $L^2(\Omega)$ while the versions $L^2_{s/a}(\Omega)$ are used if one implements exchange symmetry between the two particles. For example, in \cite{KernerBoundROMP} the authors considered a pair of two electrons with same spin which implies that one has to work on $L^2_a(\Omega)$ since electrons are fermions. Contrary to this, in \cite{KernerPairsElectrons} the electrons were assumed to be of opposite spin which then requires to work on $L^2_s(\Omega)$. 

Now, on any of those Hilbert spaces, $H$ is rigorously realised via its associated (quadratic) form, $j \in \{0,s,a\}$, 
\begin{equation}
q_{j}[\varphi]:=\int_{\Omega}|\nabla \varphi|^2\ \ud x
\end{equation}
with form domain $\cD_{j}=\{\varphi \in H^1(\Omega):\varphi \in L^2_{j}(\Omega)\ \text{and}\ \varphi|_{\partial \Omega_D}=0 \}$ where
\begin{equation}
\partial \Omega_D:=\{(x,y)\in \mathbb{R}^2_+: |x-y|=d \}\ .
\end{equation}
 Note here that $q_j[\cdot]$ obviously is a closed positive form with a dense form domain \cite{BEH08}. Also note that we write $H_j$ for the realisation of $H$ associated with the corresponding form $q_j[\cdot]$. 
\begin{remark} It is clear that the self-adjoint operator associated with $q_j[\cdot]$ is nothing else than a version of the two-dimensional Laplacian $-\Delta$ \cite{GilTru83}.
\end{remark}
In order to formulate our main result we recall the following statement which was proved in \cite{KernerBoundSystem,KernerBoundROMP,KernerPairsElectrons}. We denote by $\sigma_{\mathrm{d}}(\cdot)$ the discrete spectrum. 
\begin{prop}\label{PropI} For every $j \in \{0,s,a\}$ one has $$ \sigma_{\mathrm{d}}(H_j) \neq \emptyset\ . $$
\end{prop}
It is our goal in this note to prove the following result. 
\begin{theorem}\label{MainResult} For every $j \in \{0,s,a\}$ one has
	$$ \sigma_{\mathrm{d}}(H_j)=\{E_j  \}\  $$
	with some $E_j \geq 0$ which is an eigenvalue of multiplicity one. In other words, the discrete spectrum consists of one eigenvalue only.
\end{theorem}
\section{Proof of Theorem~\ref{MainResult}}
In this section we establish a proof of Theorem~\ref{MainResult}. We note that this was already proved for $j=a$ in \cite{KernerBoundROMP} but for the sake of completeness we also include a proof thereof.

In a first step we establish an auxiliary result: We define the domain
\begin{equation}
\widetilde{\Omega}:=\{(x,y) \in \mathbb{R}^2: |x-y| < d \}\  \cup \ \{(x,y) \in \mathbb{R}^2: |x+y| < d  \}\ 
\end{equation}
and introduce on $L^2(\widetilde{\Omega})$ the two-dimensional Laplacian $-\Delta$ with Dirichlet boundary conditions along $\partial \widetilde{\Omega}$. We denote this operator by $-\widetilde{\Delta}^{(D)}_d$. Note that the quadratic form associated with $-\widetilde{\Delta}^{(D)}_d$ is given by
\begin{equation*}
\widetilde{q}_d[\varphi]:=\int_{\widetilde{\Omega}}|\nabla \varphi|^2\ \ud x
\end{equation*}
with form domain $\widetilde{\cD}_d:=\{\varphi \in H^1(\widetilde{\Omega}): \varphi|_{\partial \widetilde{\Omega}}=0  \}$. 
\begin{prop}\label{AuxRes} The discrete spectrum of the self-adjoint operator $-\widetilde{\Delta}^{(D)}_d$ consists of exactly one eigenvalue with multiplicity one.
\end{prop}
\begin{proof} Since the Laplacian is invariant under rotations as well as translations, we may prove the statement by considering the Dirichlet Laplacian on a rotated version of $\widetilde{\Omega}$. Namely, we consider the Dirichlet Laplacian on the ``cross-shaped'' domain 
	\begin{equation*}\begin{split}
	\Omega_0:=&\{(x,y) \in \mathbb{R}^2:\ -\infty < y < \infty \ , -d/\sqrt{2}< x < +d/\sqrt{2} \} \\
	&\quad \cup \{(x,y) \in \mathbb{R}^2:\ -\infty < x < \infty \ , -d/\sqrt{2}< y < +d/\sqrt{2} \}\ .
	\end{split}
	\end{equation*}
We denote this operator by $-\Delta^{(0)}_D$. We then employ a bracketing argument and for this we introduce the direct sum of Laplacians $-\Delta_1 \oplus -\Delta_2 $; here $-\Delta_1$ is the two-dimensional Laplacian defined on the bounded domain (square)
\begin{equation*}
\Omega_1:= (-d/\sqrt{2},+d/\sqrt{2}) \times (-
d/\sqrt{2},+d/\sqrt{2})
\end{equation*}
with Neumann boundary conditions along $\partial \Omega_1$. $-\Delta_2$ denotes the two-dimensional Laplacian on the domain 
\begin{equation*}
\Omega_2:= \overset{\circ}{(\Omega_0 \setminus \Omega_1)}\
\end{equation*}
with Neumann boundary conditions along the boundary segments adjacent to $\Omega_1$ and Dirichlet boundary conditions elsewhere. Most importantly, in terms of operators we obtain the inequality
\begin{equation*}
-\Delta_1 \oplus -\Delta_2 \leq -\Delta^{(0)}_D 
\end{equation*}
which implies $N(-\Delta^{(0)}_D, E) \leq N(-\Delta_1 \oplus -\Delta_2,E)$ with $N(\cdot,E)$ denoting the counting function that counts the number of eigenvalues up to energy  $E < \inf \sigma_{\mathrm{ess}}(-\Delta_1 \oplus -\Delta_2)$. Here $\sigma_{\mathrm{ess}}(\cdot)$ denotes the essential spectrum.

From the definition of $\Omega_2$ it readily follows that
\begin{equation*}
N(-\Delta_1 \oplus -\Delta_2,E)=N(-\Delta_1,E)
\end{equation*}
whenever $E < \inf \sigma_{\mathrm{ess}}(-\Delta_1 \oplus -\Delta_2)$. The reason for this is that $\inf \sigma_{\mathrm{ess}}(-\Delta_1 \oplus -\Delta_2)=\inf \sigma_{\mathrm{ess}}(-\Delta_2)=\inf \sigma(-\Delta_2)=\frac{\pi^2}{2d^2}$, see \cite{ExnerLXXX,KernerBoundROMP} for more details (note that the spectrum of $-\Delta_1$ is purely discrete and $\Omega_2$ consists of four rectangular parts for which a separation of variables can be employed to determine the (essential) spectrum directly).

Now, in order to study $N(-\Delta_1,E)$ we take advantage of the fact that $\Omega_1$ is a square. Hence, we can employ a separation of variables which allows us to determine the eigenvalues of $-\Delta_1$ explicitly. Namley,
 \begin{equation*}
 \sigma_{\mathrm{d}}(-\Delta_1)=\left\{0,\frac{\pi^2}{2d^2},\frac{\pi^2}{d^2},... \right\}\ .
 \end{equation*}
Hence, it follows that $N(-\Delta_1,E)=1$ for all $E < \frac{\pi^2}{2d^2}$ which proves the statement taking into account that $\inf \sigma_{\mathrm{ess}}(-\Delta^{(0)}_D)=\frac{\pi^2}{2d^2}$, see
\cite{KernerBoundROMP}.
\end{proof}

\begin{proof}(of Theorem~\ref{MainResult}) We first consider the cases $j \in \{0,s\}$: We introduce the (injective) linear map 
	\begin{equation*}
	I_j:\cD_j \rightarrow \widetilde{\cD}_d \ ,
	\end{equation*}
	where $I_j\varphi$ is constructed as follows: one takes $\varphi \in \cD_j$ and then reflects it across the $y$-axis. This new function (consisting of the original $\varphi$ and the new reflected part) is then reflected another time across the $x$-axis, finally yielding an element of $\widetilde{\cD}_d$. Now, from the min-max principle we then conclude that, $n\in \mathbb{N}$ ,
	\begin{equation*}\begin{split}
	\mu_n(H_j)&=\inf_{W_n \subset \cD_j}\sup_{0\neq \varphi \in W_n}\frac{q_j[\varphi]}{\|\varphi\|^2_{L^2(\Omega)}}\\ 
	&\geq \inf_{W_n \subset I_j\cD_j}\sup_{0\neq I_j\varphi \in W_n}\frac{\widetilde{q}_d[I_j\varphi]}{\|I_j\varphi\|^2_{L^2(\widetilde{\Omega})}} \\
	&\geq \inf_{W_n \subset \widetilde{\cD}_d}\sup_{0\neq \varphi \in W_n}\frac{\widetilde{q}_d[\varphi]}{\|\varphi\|^2_{L^2(\widetilde{\Omega})}}=\mu_n(-\widetilde{\Delta}^{(D)}_d)\ ,
	\end{split}
	\end{equation*}
	where $\mu_n(\cdot)$ denotes the $n$-th min-max ``eigenvalue'' \cite{BEH08,NonnenmacherSkript}. Also, $W_n$ refers to $n$-dimensional subspaces. 
	From Proposition~\ref{AuxRes} it follows that $\mu_1(-\widetilde{\Delta}^{(D)}_d)$ is the only eigenvalue in the discrete spectrum and $\mu_n(-\widetilde{\Delta}^{(D)}_d)=\inf \sigma_{\mathrm{ess}}(-\widetilde{\Delta}^{(D)}_d)=\frac{\pi^2}{2d^2}$ for $n > 1$. Hence, by Proposition~\ref{PropI} we conclude that only $\mu_1(H_j) < \frac{\pi^2}{2d^2}$ which yields the statement since both essential spectra start at $\frac{\pi^2}{2d^2}$ as shown in \cite{KernerBoundSystem,KernerPairsElectrons}.
	
	In a next step we consider the case $j=a$: Again we want to make use of the min-max principle and hence introduce the (injective) linear map
	\begin{equation*}
	I_a:\cD_a \rightarrow \widetilde{\cD}_{d/2}
	\end{equation*}
	which acts as follows: to obtain $I_a\varphi$ one first restricts $\varphi \in \cD_a$ to
	\begin{equation}
	\{(x,y) \in \mathbb{R}^2_+:\ |x-y| < d \ \text{and}\ y > x  \}\ .
	\end{equation}
	This restriction is then reflected across the axis $x=0$ and then both segments are translated in the negative $y$-direction by $d/2$. Finally, we can extend this (translated) function by zero to obtain an element of $\widetilde{\cD}_{d/2}$. Now, employing the min-max principle shows that
		\begin{equation*}\begin{split}
		\mu_n(H_a)&=\inf_{W_n \subset \cD_a}\sup_{0\neq \varphi \in W_n}\frac{q_a[\varphi]}{\|\varphi\|^2_{L^2(\Omega)}}\\ 
		&\geq \inf_{W_n \subset I_a\cD_a}\sup_{0\neq I_a\varphi \in W_n}\frac{\widetilde{q}_d[I_a\varphi]}{\|I_a\varphi\|^2_{L^2(\widetilde{\Omega})}} \\
		&\geq \inf_{W_n \subset \widetilde{\cD}}\sup_{0\neq \varphi \in W_n}\frac{\widetilde{q}_d[\varphi]}{\|\varphi\|^2_{L^2(\widetilde{\Omega})}}=\mu_n(-\widetilde{\Delta}^{(D)}_{d/2})\ .
		\end{split}
		\end{equation*}
		For $n > 1$, $\mu_n(-\widetilde{\Delta}^{(D)}_{d/2})=\inf \sigma_{\mathrm{ess}}(-\widetilde{\Delta}^{(D)}_{d/2})=\frac{2\pi^2}{d^2}$ and since it was shown in \cite{KernerBoundROMP} that also $\inf \sigma_{\mathrm{ess}}(H_a)=\frac{2\pi^2}{d^2}$, we conclude the statement.
\end{proof}

\vspace*{0.5cm}

\subsection*{Acknowledgement}{} It is a great pleasure to thank S.~Egger (Haifa) for useful remarks and interesting discussions.

\vspace*{0.5cm}

{\small
\bibliographystyle{amsalpha}
\bibliography{Literature}}

\def\cprime{$'$} \def\polhk#1{\setbox0=\hbox{#1}{\ooalign{\hidewidth
  \lower1.5ex\hbox{`}\hidewidth\crcr\unhbox0}}}
\providecommand{\bysame}{\leavevmode\hbox to3em{\hrulefill}\thinspace}
\providecommand{\MR}{\relax\ifhmode\unskip\space\fi MR }
\providecommand{\MRhref}[2]{%
  \href{http://www.ams.org/mathscinet-getitem?mr=#1}{#2}
}
\providecommand{\href}[2]{#2}
\begin{thebibliography}{BHE08}

\bibitem[BCS57]{BCSI}
J.~Bardeen, L.~N. Cooper, and J.~R. Schrieffer, \emph{Theory of
  superconductivity}, Phys. Rev. \textbf{108} (1957), 1175--1204.

\bibitem[BHE08]{BEH08}
J.~Blank, M.~Havli\v{c}ek, and P.~Exner, \emph{Hilbert space operators in
  quantum physics}, Springer, 2008.

\bibitem[BK13a]{BKSingular}
J.~Bolte and J.~Kerner, \emph{Quantum graphs with singular two-particle
  interactions}, J. Phys. A \textbf{46} (2013), 045206.

\bibitem[BK13b]{BKContact}
\bysame, \emph{Quantum graphs with two-particle contact interactions}, J. Phys.
  A \textbf{46} (2013), 045207.

\bibitem[Coo56]{CooperBoundElectron}
L.~N. Cooper, \emph{Bound electron pairs in a degenerate {F}ermi gas}, Phys.
  Rev. \textbf{104} (1956), 1189--1190.

\bibitem[Evv89]{ExnerLXXX}
P.~Exner, P.~\v{S}eba, and P.~\v{S}\v{t}ov\'{\i}\v{c}ek, \emph{On existence of
  a bound state in an {L}-shaped waveguide}, Czechoslovak Journal of Physics B
  (1989), 1181--1191.

\bibitem[Gam04]{ThierryGamarchi}
T.~Gamarchi, \emph{{Quantum Physics in One Dimension}}, Oxford Unversity Press,
  2004.

\bibitem[Gia16]{GIAMARCHI2016322}
T.~Giamarchi, \emph{One-dimensional physics in the 21st century}, Comptes
  Rendus Physique \textbf{17} (2016), no.~3, 322 -- 331.

\bibitem[GT83]{GilTru83}
D.~Gilbarg and N.~S. Trudinger, \emph{Elliptic partial differential equations
  of second order}, Springer, 1983.

\bibitem[Ker]{KernerBoundROMP}
J.~Kerner, \emph{{On bound electron pairs in a quantum wire }}, preprint,
  arXiv:1708.03753, to appear in Reports on Mathematical Physics.

\bibitem[Ker18]{KernerPairsElectrons}
\bysame, \emph{On pairs of interacting electrons in a quantum wire}, Journal of
  Mathematical Physics \textbf{59} (2018), no.~6, 063504.

\bibitem[KM17]{KernerBoundSystem}
J.~Kerner and T.~Mühlenbruch, \emph{On a two-particle bound system on the
  half-line}, Reports on Mathematical Physics \textbf{80} (2017), no.~2, 143 --
  151.

\bibitem[Non]{NonnenmacherSkript}
S.~Nonnemacher, \emph{{Lecture notes: Spectral theory of selfadjoint
  operators}}, Lectures at University of Cardiff:
  https://cardiffmicrolocal.files.wordpress.com/2017/06/stephane-spectral-theory.pdf.

\bibitem[QU16]{Queisserunruh}
F.~Queisser and W.~G. Unruh, \emph{Long-lived resonances at mirrors}, Phys.
  Rev. D \textbf{94} (2016), 116018.

\bibitem[SEP]{KernerEggerPankrashkin}
J.~Kerner S.~Egger and K.~Pankrashkin, \emph{{Bound states of a pair of
  particles on the half-line with a general interaction potential}}, preprint,
  arXiv:1812.06500.

\end{thebibliography}

\end{document}